\documentclass[copyright,creativecommons]{eptcs}
\providecommand{\event}{FICS 2024} 

\usepackage{amsmath,amsthm,amssymb,latexsym}
\usepackage{graphicx}
\usepackage{mathtools}
\usepackage{xcolor} 
\usepackage{xspace} 
\usepackage{xargs} 
\usepackage{ebproof} 
\usepackage[capitalise]{cleveref}
\usepackage{thm-restate}

\theoremstyle{plain}
\newtheorem{theorem}{Theorem}
\newtheorem{lemma}[theorem]{Lemma}

\newtheorem{proposition}[theorem]{Proposition}

\newtheorem{app-theorem}{Theorem}[section]
\newtheorem{app-lemma}[app-theorem]{Lemma}
\newtheorem{app-proposition}[app-theorem]{Proposition}

\theoremstyle{definition}
\newtheorem{definition}[theorem]{Definition}

\usepackage{iftex}
\ifpdf
  \usepackage{underscore}         
  \usepackage[T1]{fontenc}        
\else
  \usepackage{breakurl}           
\fi

\NewDocumentCommand{\LM}{s}{\mathcal{L}_\mu\IfBooleanT{#1}{^-}}

\newcommand{\Var}{\mathsf{Var}}
\newcommand{\Prop}{\mathsf{Prop}}
\newcommand{\FL}{\mathsf{Clos}}
\NewDocumentCommand{\SC}{sO{\LM*}}{\Sigma\IfBooleanT{#1}{^c}[#2]}
\newcommand{\QF}[1][\LM*]{\mathsf{QF}[#1]} %
\renewcommand{\neg}[1]{\overline{#1}}
\newcommand{\conj}{{\textstyle\bigwedge}}
\newcommand{\disj}{{\textstyle\bigvee}}
\newcommand{\sq}{\mathop{\square}}
\newcommand{\dia}{\mathop{\diamondsuit}}
\newcommand{\Lit}{\mathsf{Lit}}

\NewDocumentCommand{\CO}{om}{\mathsf{CO}\IfValueT{#1}{_{#1}}(#2)}
\newcommand{\pred}[1]{\mathsf{p}(#1)}

\NewDocumentCommand{\setof}{mo}{\{ \, \IfValueTF{#2}{#1\mid #2 } {#1} \, \} }

\NewDocumentCommand{\Setof}{mo}{\bigl\{ \, \IfValueTF{#2}{#1\mid #2} {#1} \, \bigr\} }
\newcommand{\card}[1]{\lvert{#1}\rvert}

\newcommand{\lts}[1]{\mathcal{#1}}
\renewcommand{\S}{\lts{S}}

\newcommand{\den}[1]{\lVert{#1}\rVert}

\newcommand{\val}[1]{\mathcal{#1}}
\newcommand{\V}{\val{V}}

\let\phi\varphi

\title{The Limit of Recursion in State-based Systems\thanks{This work was supported by the Knut and Alice Wallenberg Foundation [2020.0199], Swedish Research Council [2017-05111] and Dutch Research Council [OCENW.M20.048]}}

\author{
Bahareh Afshari\qquad\qquad Giacomo Barlucchi \qquad\qquad Graham E.~Leigh
\institute{Department of Philosophy, Linguistics and Theory of Science\\ University of Gothenburg, Gothenburg, Sweden}
\email{\quad bahareh.afshari@gu.se\qquad\qquad giacomo.barlucchi@gu.se \qquad\qquad graham.leigh@gu.se}
}
\def\titlerunning{The Limit of Recursion in State-based Systems}
\def\authorrunning{B.~Afshari, G.~Barlucchi \& G.E.~Leigh}
\begin{document}
\maketitle
\begin{abstract}
 We prove that $\omega^2$ strictly bounds the iterations required for modal  definable functions to reach a fixed point across all countable structures. 
 The result corrects and extends the previously claimed result by the first and third authors on closure ordinals of the alternation-free $\mu$-calculus in \cite{AfshariL13}. 
 The new approach sees a reincarnation of Kozen's well-annotations, devised for showing the finite model property for the modal $\mu$-calculus. 
 We develop a theory of `conservative' well-annotations where minimality of annotations is guaranteed, and isolate parts of the structure that locally determine the closure ordinal of relevant formulas. 
 This adoption of well-annotations enables a direct and clear pumping process that rules out closure ordinals between $\omega^2$ and the limit of countability.
\end{abstract}
\section{Introduction}\label{introduction}
%
State-based systems and processes lay at the heart of computer science. Abstractly, they are no more than directed graphs, also known as Kripke frames, with states as vertices and state transitions as edges. Taking a transition from one state to another can model a step in computation and, doing so recursively, singles out computation paths through the system. To specify and verify properties of computation, in a fully abstract manner, temporal logics offer an elegant framework. 

Syntactically simple and algorithmically rich, temporal logics have been heavily studied.  Amongst them the modal $\mu$-calculus holds a special place, providing a level of abstraction that is mathematically appealing while computationally well-behaved. This logic not only subsumes well known temporal logics (LTL, CTL, PDL, as well as extensions such as $\mu$LTL and CTL*), it can also be enriched to capture, for example, probabilistic properties \cite{Nir15,LiuS15,Mio17}, hyperproperties \cite{GutsfeldMO21}, higher-dimensional properties \cite{Otto99} (see also \cite{Lange15}), and properties of higher-order recursive schemes \cite{Ong06,Ong21}.
In other words, $\mu$-calculus is a cornerstone in the mosaic of logics in computer science. 

Modal $\mu$-calculus is the extension of basic modal logic with least ($\mu$) and greatest ($\nu$) fixed point operators. Over a Kripke frame $\mathcal S$, the formula $\mu x\,\phi(x)$ is interpreted as the least fixed point of the induced monotone function $f\colon U\mapsto \phi(U)$ which maps a set of states $U$ (in $\mathcal S$) to the denotation of $\phi(x)$ modulo interpretation of $x$ as $U$. 
This fixed point can be obtained as the limit of transfinite iterations of $f$. Starting with the empty set, applications of $f$ give rise to an increasing sequence of sets of states, 
\[\emptyset\subseteq f(\emptyset)\subseteq f(f(\emptyset))\subseteq \cdots \subseteq f^\alpha(\emptyset)\subseteq f^{\alpha+1}(\emptyset)\subseteq\cdots\]
which necessarily stabilises at some ordinal: $f^{\kappa+1}(\emptyset)=f^\kappa(\emptyset)$.
The least such $\kappa$ is the closure ordinal of $f$ in $\S$.
One way to define a notion of \emph{closure ordinal} for a formula $\mu x \,\phi$ is as the supremum of closure ordinals of the induced function across all frames.

In this paper we study closure ordinals of  the $\Sigma$-fragment: formulas generated from closed $\mu$-calculus formulas and variables through the logical and modal operators, and the least fixed point operator $\mu$.
From an algebraic perspective the fragment corresponds to functions definable in the modal algebra with (definable) parameters. A more general class, amounting to the full calculus, is to admit arbitrary definable functions, including those defined through co-recursion.

There are several problems concerning closure ordinals which, to date, have only partial solutions. 
As there are countably many formulas of $\mu$-calculus there are countable ordinals that are not closure ordinals.
So, \emph{which} ordinals are closure ordinals? Aside from \emph{existence} is the question of \emph{decidability}: Is there an algorithm that decides whether any given formula has a closure ordinal?
And not least is the question of \emph{limitedness}: Can a non-trivial limit on closure ordinals be determined?
%
\subsection{Related work}
Closure ordinals have been considered only by a handful of authors.  
Most notable is the Czarnecki formulas \cite{Cz10}, simple formulas in the $\Sigma$-fragment demonstrating that every ordinal below $\omega^2$ is a closure ordinal. 
Czarnecki's formulas indicate a connection between syntactic and semantic complexity that was generalised in \cite{AfshariL13}: consider formulas of the form  $\mu x\,\phi$ with $\phi$ given by
\begin{gather*}
    \phi = (p_1\land\sq q_1\land\bigcirc_1 x)\lor(p_2\land\sq q_2\land\bigcirc_2 x)\lor\cdots\lor(p_{n}\land\sq q_n\land\bigcirc_n x)\lor\sq\bot
\end{gather*}
where $p_i$ and $q_i$ are conjunctions of literals and $\bigcirc_i\in\setof{\dia,\sq}$ for each $i$.
It is not difficult to prove that if such a formula has a closure ordinal $\alpha$, then $\alpha<\omega.(n+1)$. In \cite{AfshariL13}, the authors also provide a tableaux-based characterisation for the closure ordinals of the alternation-free $\mu$-calculus. In \cite{Kretz06}, Kretz proves that every valid $\Sigma_1$-formula in the one-variable fragment, which includes any valid formula of the form above, has finite closure ordinal.

Fontaine \cite{Fon08} (see also \cite{Fon10}) carries out a  study of closure ordinals of the continuous $\mu$-calculus, that is the fragment constituting formulas $\mu x\,\phi(x)$ where $\phi(x)$  is continuous with respect to $x$ in the Scott topology on the powerset algebra. It is shown the $\{\sq,\nu\}$-free fragment of $\mu$-calculus characterises the continuous $\mu$-calculus establishing that closure ordinals are obtained in at most $\omega$ iterations.
Fontaine and Venema provide syntactic characterisations of several other semantic properties in \cite{FV18}.
Gouveia and Santocanale~\cite{GS18} study $\kappa$-continuity for $\kappa$ an infinite (regular) cardinal, and prove a generalisation of the aforementioned results regarding existence: any ordinal obtainable from $0$, $1$, $\omega$, and $\omega_1$ (least uncountable ordinal) by the binary ordinal sum operation is the closure ordinal of a $\mu$-calculus formula.

One may argue that the concept of closure ordinal and questions posed about it stand somewhat remote to other investigations concerning $\mu$-calculus. But there is strong evidence that this not so. One intriguing connection, pointed out by Skrzypczak \cite{MS16}, is to the descriptive complexity of B\"uchi languages. Each B\"uchi automaton can be associated a rank below $\omega_1$, measuring the complexity of the automaton against input trees. It is shown that the rank of an automaton $\mathcal B$ is strictly below $\omega_1$ if and only if the language of $\mathcal B$ is Borel, and strictly below $\omega^2$ if and only if the language is weak monadic second order definable. Skrzypczak proposes that the pumping arguments central to deducing bounds on closure ordinals may be used to tackle questions of definability (and decidability) of non-deterministic languages, the so-called \emph{gap properties} for B\"uchi languages (see e.g.~\cite{SW16,NW03}).
 
Milanese \cite{Milanese18} studies closure ordinals over bidirectional models and shows that every ordinal below $\omega^\omega$ is a closure ordinal  (see also \cite{MilaneseV19}). This result was observed independently in  \cite{AJL19} as part of a study of $\omega$-branching proof systems for the two-way $\mu$-calculus. Again, such results add weight to the claim that closure ordinals are entwined in many topics concerning $\mu$-calculi.
\subsection{Contribution}
We prove that $\omega^2$ is a strict upper bound on the closure ordinals of the $\Sigma$-fragment of the modal $\mu$-calculus. This reproves and extends the claims in \cite{AfshariL13} concerning closure ordinals of the alternation-free fragment.
There are two critical errors in \cite{AfshariL13}, both in the original proof of the `pumping' Lemma 3.18.\footnote{The authors are indebted to Micha\l\ Skrzypczak and Igor Walukiewicz for identifying one of the errors in \cite{AfshariL13}.}
While the errors can be fixed (see unpublished notes \cite{AL16}) it is at the cost of a weaker result and a more technically involved argument that appeals, in particular, to the closure ordinals of valid $\Sigma_1$-formulas.

The approach presented here develops a theory of ordinal annotations that simplifies the conceptual framework compared to~\cite{AL16} and lays the groundwork for future extensions to the full calculus.
At its base is the notion of well-annotations, employed by Kozen to establish the finite model property for $\mu$-calculus~\cite{Koz88}.
We refine the concept by imposing constraints on the annotating ordinals so that the existence of such a `conservative' well-annotation corresponds to the existence of certain closure ordinals. 
The theory of well-annotations becomes more tractable by also restricting the underlying syntax.
Rather than the traditional syntax of the modal $\mu$-calculus, we consider formulas constructed via a single modal operator -- related to the `cover' modality of \cite{JW95} -- and present them as modal equation systems in conjunctive form. 
It is shown that both expressivity and bounds on closure ordinals are preserved through this syntactic preprocessing.
The central argument involves  a pumping lemma for well-annotations.
Assuming the existence of a sufficiently `large' conservative well-annotation, a transfinite series of substitutions shows it possible to obtain a conservative well-annotation corresponding in size to an arbitrary countable ordinal, thereby refuting the existence of closure ordinals  equal or greater than \( \omega^2 \). 

\section{Modal  \texorpdfstring{$\mu$}{mu}-calculus}\label{sec:mu}

We adopt a unimodal presentation of modal logic eschewing the usual unary modal operators \( \sq \) and \( \dia \) for a single modality \( \nabla \) that takes finitely many formulas as arguments.
The intended interpretation of \( \nabla \Gamma \) in terms of \( \sq \)/\( \dia \) syntax is \( \disj_{\gamma \in \Gamma} \sq \gamma \lor \dia\! \conj \Gamma \).
This modality is the classical dual of the ‘cover’ modality originally introduced by Janin and Walukiewicz~\cite{JW95} and has proved especially well suited for investigating the modal and co-algebraic logics~\cite{JW95,Wal00, Kupke12,KupkeP11}.
The formulas of the modal \( \mu \)-calculus, denoted \( \LM \), are those generated by the following grammar. 
\begin{align*}
	\phi &\coloneqq 
		p \mid 
		\neg{p} \mid 
		x \mid 
		\disj \Gamma \mid 
		\conj \Gamma \mid 
		\nabla \Gamma  \mid 
		\mu x\, \phi \mid 
		\nu x\, \phi
	\\
	\Gamma &\coloneqq 
		\emptyset \mid 
		\Gamma \cup \setof{ \phi }
\end{align*} 
where $x$ and $p$ range over, respectively, a set \( \Var \) of \emph{variables} and \( \Prop \) of \emph{propositional constants}.
Note, negation is not included as a logical connective, except for propositional constants, expressed by the atoms \( \neg p \) above. 
Non-variable atoms, namely propositional constants and their negations, are called \emph{literals}, the set of which is denoted \( \Lit \).

We utilise abbreviations 
$ \bot \coloneqq \disj \emptyset $ and $ \top \coloneqq \conj \emptyset $, and represent binary conjunction and disjunction via, respectively,
\( \phi_0 \land \phi_1 \coloneqq \conj \setof{ \phi_0 , \phi_1 }\) and \( \phi_0 \lor \phi_1  \coloneqq \disj \setof{ \phi_0 ,  \phi_1 } \).
With the intended interpretation of \( \nabla \), the two unary modalities $\sq$ and $\dia$ are  recovered by
\(
	\sq \phi \coloneqq \nabla\setof{ \phi,\bot } \) and \(
	\dia\! \phi \coloneqq \nabla\setof{ \phi } \land \nabla \emptyset
\).

Free and bound variables are defined per usual.
A formula with no free variable occurrences is called \emph{closed} and
we write  \( \LM* \) for the set of closed formulas.
For a set of formulas $F$, the \emph{quantifier-free formulas over $F$}, denoted $\QF[F]$, is  the closure of $ F \cup \Lit$ under the logical connectives and  $\nabla$.

Formulas are interpreted with respect to Kripke frames.	A \emph{frame} is a tuple $\S = (S,R,\Lambda)$ comprising a non-empty set $S$ of \emph{states}, a binary \emph{accessibility} relation $R \subseteq S\times S$ and a \emph{labelling} function $\Lambda \colon \Prop \to \mathcal P ( S )$ from propositional constants to sets of states. A frame is often identified with its set of states.
For a frame \( ( S , R , \Lambda) \) and $s\in S$, we write  $R [s]$ for  the set of successors of $s$, namely, $\setof{ t\in S }[ (s,t) \in R ]$.
Given a formula, a frame $\S = (S,R,\Lambda)$ and a valuation function $\V \colon \Var \to \mathcal{P}(S)$, the denotation of \( \phi \) in \( \S \) relative to \( \V \) is the set \( \den \phi ^\S _\V \) defined by
\begin{align*}
	\den x ^\S_\V &= 
		\V (x)
		& 
	\den { \conj \Gamma } ^\S_\V &= 
		\bigcap \Setof{ \den \gamma ^\S_\V }[ \gamma \in \Gamma ]
	&
	\den { \nu x\, \phi}^\S_\V &=
		\bigcup \Setof{
			U \subseteq S 
		}[
			U \subseteq \den \phi ^{\S}_{\V [x\mapsto U]} 
		]
	\\
	\den p ^\S_\V &= 
		\Lambda(p)
	&
	\den { \disj \Gamma } ^\S_\V &= 
		\bigcup \Setof{ \den \gamma ^\S_\V }[ \gamma \in \Gamma ]
	&
	\den {\mu x\, \phi} ^\S_\V &=
		\bigcap \Setof{ 
			U \subseteq S 
		}[
			\den \phi ^{\S}_{\V [x\mapsto U]} \subseteq U
		]
	\\
	\den {\neg p} ^\S_\V &=
		\mathrlap{
		S \setminus \Lambda(p)
		}
	&& &
	\den{ \nabla \Gamma }^\S_\V 
	&= \nabla \Setof{ \den{\phi}^\S_\V }[ \phi \in \Gamma ]
\end{align*}
where \( \V[x\mapsto U] \) expresses the valuation that maps \( x \) to \( U \) and otherwise agrees with \( \V \). The function \( \nabla \colon \mathcal{P}(\mathcal{P}(S)) \to \mathcal{P}(S) \) is specified by
\begin{equation*}
	\label{eqn-semantics-nabla}
	\nabla \mathcal U \coloneqq
	\Setof{ 
		v \in S 
	}[
	     R[v] \subseteq U \text{ for some } U \in \mathcal U
	]
	\cup 
	\Setof{ 
		v \in S
	}[
		\bigcap \mathcal U \cap R[v] \neq \emptyset
    ].
\end{equation*}

\subsection{Equational formulas}
\label{s-equation-systems}

We will be working with the {$\Sigma$-fragment} of $\LM$ which, loosely speaking, consists of formulas wherein the external $\mu$-quantifiers do not bind variables in the scope of other quantifiers.
More precisely, the \emph{$\Sigma$-fragment} is the closure of $\LM* \cup \Lit \cup \Var $ under the logical connectives,  $\nabla$-modality and the  \( \mu \)-quantification. We refer to formulas of the $\Sigma$-fragment as \emph{$\Sigma$-formulas}.

For the analysis we adopt a representation of $\Sigma$-formulas based on \emph{modal equation systems}.
A modal equation system (m.e.s.)\ consists of a finite set of {equations} between variables and quantifier-free formulas, accompanied by a `priority' order on variables. 
We spare the general definition of m.e.s.\ (for which the reader can consult, e.g., \cite[sec.~8.3.4]{DGL16}) and focus on a formulation corresponding to $\Sigma$-formulas.
In particular, in our set-up there is no order imposed on variables, and equations relate each variable to a  quantifier-free formula over $\LM*$ and variables, that is, $\QF[\LM* \cup \Var]$.
\begin{definition}\label{MES}
     An \emph{equation system} over \( \LM* \) is a tuple \( ( X , E ) \) where \( X \subseteq \Var \) is a finite set of variables and  \( E \colon X \to \QF[X\cup\LM*] \) is such that all variables in \( E(x) \) are in the scope of a modality.
    An \emph{equational formula} (over \( \LM* \)) is a triple \( ( X , x_0 , E ) \) where \( ( X , E ) \) is an equation system (over \( \LM* \)) and \( x_0 \in X \) is a distinguished variable called the \emph{initial variable}.
\end{definition} 
The intended semantics of an equational formula is the denotation of the initial variable relative to the system's  equations taken under a least fixed point reading.
The formal semantics is most easily given through \emph{approximations}:

\begin{definition}[Approximations]\label{a2}
Fix an equation system \( ( X , E ) \) and frame \( \S \).
For each ordinal \( \alpha \) define a valuation \( \V^\alpha \) by 
\[ 
	\V^\alpha (x) =
	\begin{cases}
		\bigcup_{\beta < \alpha} \den{ E(x) }^\S_{\V^\beta}, & x \in X,
		\\	
		\emptyset, &x \in \Var \setminus X.
	\end{cases}
\]
For each formula \( \psi \in \QF[ X \cup \LM* ]\), the \( \alpha \)-th \emph{approximation} of \( \psi \) (relative to \( (X,E ) \)), also referred to as the \emph{denotation} of $\psi^\alpha$, is \( \den { \psi^\alpha }^\S \coloneqq \den \psi ^\S_{\V^\alpha} \).
The denotation of the equational formula \( \phi = ( X , x , E ) \) in \( \S \) is defined as \( \den \phi^\S \coloneqq \bigcup_{\alpha<\omega_1} \den{ x^\alpha }^\S \).
\end{definition}
That every $\Sigma$-formula is equivalent to an equational formula over $\LM*$ can be shown via a simple translation between the two representations that replaces the `external' \( \mu \)-operators by equations and vice-versa.
Henceforth, we identify $\Sigma$-formulas and equational formulas.

We will utilise a special form of equational systems/formulas that facilitates the desired pumping argument while staying faithful to both expressivity and closure ordinals within the $\Sigma$-fragment. 
\begin{definition}[Conjunctive system]\label{conj-mes}
	An equation system \( ( X , E ) \) is said to be \emph{conjunctive} if for every $x\in X$, the formula $E(x)$ is  of the form
    \( \conj_{i<k} ( \disj \Gamma_{i} \lor \nabla Y_{i} )\)
    for some \( \Gamma_i \subseteq \LM* \) and \( Y_i \subseteq X \).
An equational formula over a conjunctive system is called a \emph{conjunctive formula}.
\end{definition}
One obvious constraint is that in the syntax above it is not possible to express \( \bot \) as the empty disjunction. 
Instead, \( \bot \) is expressed as the conjunctive equation \( z \mapsto \dia\! z \) where, recall, \( \dia\! z = \nabla \setof{ z } \land \nabla \emptyset \).
Note also that, in a conjunctive equation, the modal depth is trivial and a conjunct may contain at most one $\nabla$-modality. That the resulting fragment is as expressive as the $\Sigma$-fragment is essentially the dual of Janin and Walukiewicz' `disjunctive normal form' theorem~\cite{JW95}.
Less obvious is the preservation of closure ordinals which will be addressed in the next section (see \cref{t-conj-equiv}).
	
It is worth highlighting that what we have called `conjunctive' here is most correctly the `conjunctive $\Sigma$-fragment'. Since we only work with the $\Sigma$-fragment in this article we  opt for the shorter name convention.

We use the following adaptation of the standard  Fischer--Ladner closure of formulas~\cite{FischerLadner-79} to equation systems.
	The \emph{closure} of an equation system \( ( X , E ) \) is the smallest set \( \FL ({X,E}) \subseteq \LM \) satisfying 
 (1) $ E(X) \subseteq \FL ({X,E})$; 
 (2) if $\bigcirc \Gamma \in \FL ({X,E}) $ for $\bigcirc \in \{ \conj , \disj, \nabla\}$ then $ \Gamma \subseteq \FL ({X,E})$; and 
 (3) if $ \sigma x \, \psi \in \FL ({X,E}) $ for \( \sigma \in \{ \mu,\nu\} \) then $\psi(\sigma x\, \psi/x) \in \FL ({X,E}) $ where \( \psi(\chi/x) \) denotes the result of substituting \( \chi \) for free occurrences of \( x \) in \( \psi \), avoiding variable capture.
	The size of the equation system \( (X,E) \), written \(\card{X,E}\), is the cardinality of its closure.

\subsection{Closure ordinals}\label{sec:CO}

As remarked, an $\LM$-formula $\phi(x)$ considered over a frame $\S$ induces a monotone function
on the powerset lattice $(\mathcal P(S),\subseteq)$ mapping a set of states $U\subseteq \S$ to  
\(
\den { \phi }^\S_{\V[ x \mapsto U ]}
\). 
One may give an approximation semantics for \( \LM \) by iterating this function into the transfinite, where \( \Omega \) denotes the class of ordinals:
\begin{equation}
		\label{eq:mu-app}
    \den{\mu x\,\phi}^\S_\V = 
		\bigcup_{\kappa \in \Omega }\den{ \mu^\kappa x\, \phi }^\S_\V
    \quad\text{where}\quad
	    \den {\mu^\alpha x\, \phi }^\S_\V 
	    \coloneqq 
    \bigcup_{\beta<\alpha} \den \phi^\S_{\V[ x \mapsto \den {\mu^\beta x\,\phi}^\S_\V ]}
\end{equation}
The `formula' \( \mu^\kappa x\, \phi \) expresses the $\kappa$-iteration of the function \( f \) starting on \( \emptyset \).
In particular, for every $ s\in \den{\mu x \, \phi}^\S_\V$ there is some $\kappa \in \Omega $ s.t.\ $s\in \den{ \mu^{\kappa}x \, \phi }^\S_\V$.
If \( \S \) is a countable frame, cardinality considerations show that \( \den{\mu x\,\phi}^\S_\V = \den{ \mu^\kappa x\, \phi }^\S_\V \) for some \( \kappa < \omega_1 \).
Thus, for each closed \( \mu x\, \phi  \), there exists \( \kappa \le \omega_1 \) such that \( \den{\mu x\, \phi}^\S_\V = \den{\mu^\kappa x\, \phi }^\S_\V \) for \emph{every} countable frame \( \S \) and valuation \( \V \).

It is essentially these approximations that provide the semantics of equational formulas in the previous section though there is a noteworthy difference: in \cref{a2}  
the ordinal annotation adopted consists of a single ordinal number `counting' multiple variables.
Implicit in \eqref{eq:mu-app} is what is known as the  \emph{signature},
an $n$-tuple of ordinals keeping record of the iterations of each $\mu$ operator in \( \mu x \, \phi \).
In the context of an equation system \( ( X , E ) \) a signature is an assignment of an ordinal to each variable in \( X \). 
Assuming a fixed enumeration \( x_0 , \dotsc , x_n \) of \( X \), a signature is a sequence of ordinals \( \alpha_0 \dotsm \alpha_{n} \) and the denotation of quantifier-free formulas over \( X \) relative to this signature uses \( \alpha_i \) to interpret \( x_i \):
\( \den{ x_i ^{\alpha_0 \dotsm \alpha_n } }^\S = \bigcup_{\beta < \alpha_i } \den{ E(x_i)^{ \dotsm \alpha_{i-1} \beta  \alpha_{i+1} \dotsm } }^\S \)\!.
For a detailed definition and properties of signatures we refer the reader to~\cite{SE89,NW96,DGL16}.
The `single approximation' notion of denotation in \cref{a2}, which counts each and every unfolding of equations, may appear a crude measure in comparison to the fine-grained specification that treats each equation independently. 
Signatures, however, also introduce complications of a `book-keeping' nature while offering a level of detail that is not needed for characterising bounds on fixed point iterations as shown by the following lemma.
\begin{lemma}
	Let \( \alpha_0 , \dotsc, \alpha_n < \omega_1 \) and \( \alpha = \sum_{i} \alpha_i \).
	For every equational formula \( \psi \) with variables over \( x_0 , \dotsc, x_n \) and structure \( \S \), \( \den{ \psi^{\alpha_0 \dotsm \alpha_n} }^\S \subseteq \den{ \psi ^\alpha } ^\S \subseteq \den{\psi^{\alpha \dotsm \alpha}}^\S \). 
\end{lemma}
Following the notion of ordinal approximation in \cref{a2} we  define closure ordinals of $\Sigma$-formulas as follows.
\begin{definition}[Closure Ordinal]\label{CO}
	Given a frame $\S$ and $\Sigma$-formula $\phi$ presented as an equational formula, \emph{the closure ordinal of \( \phi \) in \( \S \)} is the least ordinal $\kappa = \CO[\S] \phi$ such that $\den {\phi^\kappa}^\S  = \den { \phi } ^\S$. 
	The \emph{closure ordinal} of $\phi$, denoted $ \CO \phi$, is the least ordinal $\kappa$ such that for all countable frames $\S$, $\CO[\S] \phi \le \kappa$. 
\end{definition}
As the definition above restricts attention to countable frames, every formula has closure ordinal bounded by the first uncountable ordinal \( \omega_1 \).
Cardinality considerations show that not every countable ordinal is a closure ordinal.
Yet it is open as to precisely which countable ordinals are closure ordinals.
The following partial result was established by Czarnecki~\cite{Cz10}.
\begin{proposition}
	\label{Czarnecki}
	For every \( \alpha < \omega^2 \) there exists a $\Sigma$-formula \( \phi \) such that \( \CO{ \phi } = \alpha \).
\end{proposition}
We end this section with a result on preservation of closure ordinals between equivalent formulas.
\begin{restatable}{theorem}{thmconjequiv}
	\label{t-conj-equiv}
	Let \( \phi \) be a $\Sigma$-formula. There exists a conjunctive formula \( \phi_c \) such that \( \CO {\phi } \le \CO{\phi_c} \), and  \( \CO \phi < \omega_1 \) implies \( \CO{\phi_c} < \omega_1 \).
\end{restatable}
The proof of \cref{t-conj-equiv} proceeds by converting each equational formula into an equivalent conjunctive one.
This transformation is, in essence, determinisation of alternating parity tree automata~\cite{green}. Due to space reasons, the syntactic translation from arbitrary equational formulas to conjunctive ones is not shown here. We point out, however, that for the theorem we require that the conjunctive formula so obtained preserves existence of a countable closure ordinal, a property which is not invariant under mere logical equivalence (the reader can compare formulas \( \mu x\, \top \) and \( \mu x ( \sq x \lor \nu y\, \dia\! y  ) \)).

\section{Conservative well-annotations}\label{sec:cwa}

The notion of well-annotations is taken from \cite{Koz88} with minor changes to adapt to conjunctive equation systems. 
An annotated formula is a pair \( ( \phi , \alpha ) \), written \( \phi^\alpha \), where \( \phi \) is a formula and \( \alpha < \omega_1 \).
For a set \( \Theta \)  of annotated formulas, \( \Theta^- \) denotes  the underlying formulas: \( \Theta^- \coloneqq \setof{ \phi }[ \phi^\alpha \in \Theta \text{ for some } \alpha < \omega_1 ] \).
For \( \Gamma \) a set of unannotated formulas, let \( \Gamma^\alpha = \setof{ \phi^\alpha }[\phi \in \Gamma] \).
We utilise a relation $\preceq$ on sets of annotated formulas, defined by
\(
    \Theta \preceq \Xi \) iff for all \( \phi^\alpha \in \Xi \) there exists \( \beta \le \alpha \)  such that \( \phi^\beta \in \Theta 
\).%
\begin{definition}[Well-annotation]\label{d-well-anno}
An \emph{annotation} of a frame \( \S \) for an equation system \( (X,E)\) is a function \( \Theta \colon \S \to \mathcal{P}(\LM \times \omega_1 ) \) associating to each state $s \in \S$ a set \( \Theta_s \) of annotated formulas from \(  \FL(X,E) \). A \emph{well-annotation} of $\S$ for \( (X,E) \) is an annotation \( \Theta \) such that for all \( s \in \S \), \( \phi \), \( \Gamma \) and \( \alpha < \omega_1 \):
\begin{enumerate}
	\item\label{d-well-anno-1} if $\phi^\alpha \in \Theta_s $ and \( \phi \in \LM* \) then $s \in \den \phi ^\S $;
	\item if $ x^\alpha \in \Theta_s$ for \( x \in X \), then $ E(x) ^\beta \in \Theta_s$ for some $\beta < \alpha$;
	\item if $ \disj \Gamma ^\alpha \in \Theta_s$ then $ \Gamma^{\beta} \cap \Theta_s \neq \emptyset $ for some \( \beta \le \alpha \); 
	\item if $\conj \Gamma ^\alpha \in \Theta_s$ then \( \Theta_s \preceq \Gamma^{\alpha} \);	
	\item\label{d-well-anno-modal} if $ \nabla \Gamma ^\alpha \in \Theta_s$, then one of the following properties holds 
	\begin{enumerate}
    \item there exists \( r \in R[s] \) such that \( \Theta_r \preceq \Gamma^\alpha \), 
    \item there exists $\phi \in \Gamma$ such that \( \Theta_r \preceq \setof{ \phi^\alpha } \) for all \( r \in R[s]\). 
\end{enumerate}
\end{enumerate}\ 
\end{definition}
Recall that quantified formulas in this setting are all in \(\LM*\), hence they are considered in \ref{d-well-anno-1}.
When referring to well-annotations we omit explicit mention of the underlying frame and associated equation system if there is no cause for confusion. 
The relation \(\preceq\) introduced above is extended to annotations in a pointwise manner.
That is, for annotations \( \Theta \) and \( \Xi \) of a frame \( \S \), set \( \Theta \preceq \Xi \) iff \( \Theta_s \preceq \Xi_s \) for all \( s \in \S \).
In the following, \( \S , s \vDash \Theta_s \) expresses that \( s \in \den{ \phi^\alpha }^\S \) for every \( \phi^\alpha \in \Theta_s \). 
\begin{theorem}\label{kozen} Given an annotation \( \Theta \) of \( \S \),
\begin{enumerate}
	\item If $\Theta$ is a well-annotation, then $ \S , s \vDash \Theta_s $ for every $s\in \S$.
	\item If $\S,s\vDash \Theta_s$ for all $s$, then there is a well-annotation $\Theta'$ of \( \S \) such that \( \Theta' \preceq \Theta \).
\end{enumerate}
\end{theorem}
\begin{proof}
	See \cite[Lemma 4.2]{Koz88}.
\end{proof}
We are interested in well-annotations that are \( \preceq\)-minimal for a given frame and equation system.
We call these annotations \emph{conservative}.
\begin{definition}[Conservative well-annotation]\label{conservative}
	A well-annotation $\Theta$ of \( \S \) is conservative if two conditions are met:
	\begin{enumerate}
		\item for every \( s \in \S \) and \( \phi \) there is at most one \( \alpha < \omega_1 \) such that \( \phi^\alpha \in \Theta_s \).
		\item for every well-annotation \( \Theta' \) of \( \S \), $\Theta \preceq \Theta'$.
	\end{enumerate}\ 
\end{definition}
The existence of conservative well-annotations is guaranteed by
\begin{proposition}
	\label{p-conservative}
	Let \( \S \) be a frame, \(r\in S\) and \( x \in X \) such that \( r \in \den x^\S \). 
	There exists a conservative well-annotation \( \Theta \) of \( \S \) such that \( x \in \Theta_r^- \).
\end{proposition}
\begin{proof}
	The desired annotation is given by \( \phi^\alpha \in \Theta_s \) iff \( s \in \den{\phi}^\S \) and \( \alpha \) is least such that \( s \in \den{\phi^\alpha}^\S \).
\end{proof}
The following proposition provides the crucial link between conservative well-annotations and closure ordinals.
\begin{proposition}
	\label{conservative-CO}
	Suppose \( \Theta \) is a conservative well-annotation of \( \S \) and \( \phi^\alpha \in \Theta_s \) for some \( s \in \S \).
	Then $\CO \phi\geq \alpha$.
\end{proposition}
\begin{proof}
	Let $\Theta$ be a conservative well-annotation of \( \S \) and $\phi^\alpha \in \Theta_s$. 
	Assume $\CO \phi=\gamma <\alpha$. 
	By \cref{CO},  $ s \in \den {\phi^\alpha} ^\S = \den {\phi^\gamma} ^\S$. 
	Consider the annotation \( \hat \Theta \) given by \( \hat{\Theta}_s = \Theta_s \cup \{\phi^\gamma\} \) and $\hat{\Theta}_{r} = \Theta_{r}$ for any $r \neq s$.
	\Cref{kozen} implies that \( \hat\Theta \) can be extended to a well-annotation \( \Theta' \) of \( \S \) satisfying \( \Theta'_s \preceq \hat\Theta_s \).
	But then $ \Theta \npreceq \Theta' $ contradicting the assumption that \( \Theta \) is conservative.
	We conclude that $\CO \phi \geq \alpha$.
\end{proof}

Not every formula in a conservative well-annotation of $\phi^\alpha$ plays a role in generating $\alpha$ as ordinal. 
From the large quantity of information provided by a conservative well-annotation, we want to be able to identify the part of the annotation that is relevant in determining the main ordinal, i.e., its \emph{relevant part}. 
Specifically, for a formula \( \phi^\alpha\) that is in the relevant part of some $\Theta_s$, the definition of relevant part guarantees that a new frame with \(\phi\) annotated by some ordinal \(>\alpha\) can be obtained, if it is possible to alter the initial frame in a way that (1) does not alter the formulas satisfied at the successors of \( s \), and (2) increases the ordinal annotation of every relevant formula at a successor of \( s \) to some ordinal \( > \alpha \).

For the following we introduce some notation concerning a well-annotation \( \Theta \) of \( \S \). 
Given the twofold condition on $ {\nabla \Gamma^\alpha\in\Theta_s} $ in \cref{d-well-anno}, it is useful to identify the sets witnessing the two existential claims.
For \( s \in \S \) and \( \Gamma \subseteq \LM \), define
\begin{align*}
	&\Gamma^s_{\sq} \coloneqq \setof{ \phi \in \Gamma }[ \phi \in \bigcap_{t \in R[s]} \Theta_t^- ]
	&
	&\Gamma^s_{\dia} \coloneqq \setof{ t \in R[s] }[ \Gamma \subseteq \Theta_t^- ] .
\end{align*}
For \( \alpha > 0 \), define \( \pred \alpha \) as the least ordinal \( < \alpha \) such that \( \alpha = \pred \alpha + \omega^\eta \) for some \( \eta \). Note also that \( \eta \) is uniquely determined and independent of the choice of \( \pred \alpha \).
\begin{definition}[Relevant part]\label{d-relevant}
	Let \( \Theta \) be a conservative well-annotation of \( \S \) and \( \Phi \) an annotation of the same frame.
	We call \( \Phi \) a \emph{relevant part} of \( \Theta \) if for every $s\in \S$,
	\begin{enumerate}
	\item \( \Phi_s \subseteq \Theta_s \);
    \item  if \( x^{\alpha} \in \Phi_s \) then \( E(x)^{\beta } \in \Phi_s \) where \( \alpha = \beta + 1 \);\label{d-relevant-variable}
   \item if \( \disj \Gamma^\alpha \in \Phi_s  \) and \( \alpha > 0 \) then \( \Gamma^\alpha \cap \Theta_s \subseteq \Phi_s \);
    \item if \( \conj \Gamma^\alpha \in \Phi_s  \) then \( \chi^\alpha \in \Phi_s\) for exactly one \(\chi \in \Gamma\);
	\item\label{d-relevant-nabla} if \( \nabla \Gamma^\alpha \in \Phi_s \) and \( \alpha > 0 \) then: 
	\begin{enumerate} 
		\item\label{d-relevant-nabla-a} for all \( \eta < \alpha \) and \( \phi \in \Gamma^s_{\sq} \) there is a \( r \in R[s] \) and \( \beta > \eta \) s.t.\ \( \phi^\beta \in \Phi_r \), 
		\item\label{d-relevant-nabla-b} \( \Gamma \cap \Phi_r^- \neq \emptyset \) for every \( r \in \Gamma^s_{\dia} \), and
		\item\label{d-relevant-nabla-c} for all \( r \in R[s] \) if \( \Phi_r \neq \emptyset \) then \( y^\beta \in \Phi_r \) for some \( y \in \Gamma \) and \( \beta > \pred \alpha \).
	\end{enumerate}
	\end{enumerate}\ 
\end{definition}
Formulas in \( \Phi_s \) are referred to as \emph{relevant formulas} at \( s \).
The final condition of the definition, \ref{d-relevant-nabla-c}, has the role of ensuring that the formulas relevant at a successor state sit in the same ordinal `neighbourhood'. Since \( \pred{\omega + 1} = \omega \) and \( \pred {\omega.(k+1)} = \omega.k \), in these cases all continuations through a modality should be annotated by at least \( \omega + 1 \) and \( \omega.k \) respectively.
In other words, viewing the sequences of formulas in the relevant part as a formula `trace' through the well-annotation \( \Theta \), these traces are restricted in the size of their ordinal decrements.
Notice, however, that the ordinal annotations along these relevant `traces' need not be weakly decreasing.
If for \( r \in R[s] \) we have \( \Gamma \subseteq \Theta_r^- \) then we require some \( \psi^{\beta} \in \Gamma^\beta \cap \Theta_r \) to be marked as relevant even if \( \beta > \alpha \).
The reason for this requirement is that such a successor \( r \), although not `relevant' to witnessing the ordinal of \( \nabla \Gamma \) in \( \Theta_s \) may become `relevant' after an  attempt to force an increase in the ordinal annotation.
It could be the case, for example, that for some \( r \in R[s] \) we have \( \Gamma^{\alpha+2} \subseteq \Theta_r \);
if the annotation of all relevant formulas in \( R[s]\setminus\setof r \) is increased by, say, \( \omega \) then without also increasing the annotation at \( r \) we find that \( \nabla \Gamma^\alpha \) at \( s \) has increased only to \( \nabla \Gamma^{\alpha+2} \).

In order to isolate sufficient conditions for undertaking a pumping of well-annotations, a further constraint can be placed on relevant parts to the effect that each path through the underlying frame carries at most one relevant trace of formulas.
In the context of conjunctive formulas, this condition amounts to there being at most one relevant modal formula at each state.
\begin{restatable}{proposition}{prelevantexists}
	\label{p-relevant-exist}
	Let \( \phi = ( X , x , E ) \) be a conjunctive formula.
	Let \( \Theta \) be a conservative well-annotation and \( x^\alpha \in \Theta_{s} \) for some \( \alpha < \omega_1 \) and state $s$.
	There exists a tree \( \lts T \), a conservative well-annotation \( \Theta' \) of \( \lts T \) and a relevant part \( \Phi \) satisfying:
	\begin{enumerate}
		\item \( x^\alpha \in \Phi_{r} \) where \( r \) is the root of \( \lts T \).
		\item\label{p-relevant-exist2} For every \( t \in \lts T \), there is at most one \( \Gamma \) such that \( \nabla \Gamma \in \Phi_t^- \).
	\end{enumerate}
\end{restatable}
    The crux of the argument is in ensuring that the requirements of \cref{d-relevant} can be met while marking at most one modal formula as relevant at each state. 
    The restriction to conjunctive formulas makes condition \ref{d-relevant-nabla} of the definition the only non-trivial case.
    Duplicating successor nodes enables the desired assignment of relevant formulas to successors.

\section{Limits on closure ordinals}\label{sec:bound}

 The argument showing that $\omega^2$ is an upper bound on the closure ordinals of formulas in the $\Sigma$-fragment comprises two parts.
 First, the existence of a pumping procedure for sufficiently large conservative well-annotations is established. As a consequence, for every formula $\phi$ in the fragment there is a measure $N$, related to the size of the formula, that determines an interval $[\omega.N,\omega^2)$ where the possibility of the closure ordinal of $\phi$ is excluded.
In the second part, the interval is extended to all the (countable) ordinals above $\omega.N$ by proving that the consequences of the pumping method reach beyond $\omega^2$.
Combining the two parts, we obtain $\omega^2 = \sup_k \omega.k$ as an upper bound for the $\Sigma$-fragment.
 
A \emph{path} in a  frame $\S=(S,R,\Lambda)$  is a sequence of states \( ( s_i )_{ i \le k} \) such that \( s_{i+1} \in R[s_i] \) for \( i < k \). An infinite path through \( \S \) is an infinite sequence \( ( s_i)_{i < \omega} \) such that every initial sequence is a path through \( \S \).
If $\S$ is a tree we use $\rho$ to denote the root state.
For the following let a conjunctive equation system \( ( X , E ) \) be fixed.
\begin{definition}\label{d-ordinal-ann}
Given \( \varphi, \Gamma \subseteq \FL(X,E) \), define \( O( \varphi ,\Gamma ) \) to be the supremum of ordinals \( \kappa < \omega_1 \) for which there exists a conservative well-annotation \( \Theta \) of a tree such that \( \Gamma = \Theta_\rho^- \) and \( \phi^\kappa \in \Theta_\rho \).
\end{definition}

\begin{definition}[Optimal annotation]\label{d-optimal}
Given a state $s$ in a conservative well annotation $\Theta$, $\Theta_s$ is \emph{optimal} with respect to a formula $\varphi^\alpha \in \Theta_s$ if $O(\varphi, \Theta_s) < \alpha +\omega$.
\end{definition}
Since a key element of the argument in the proof of \cref{pump} will be the non-existence of optimal paths under certain conditions, we introduce the notion of repetition pair. As recognized by \cref{optimal-trace-points}, repetition pairs are designed to entail non-optimality, given that they present candidates for additional pumping.
\begin{definition}[Repetition pair]
	\label{reptrace}
	Let \( \Theta \) be a conservative well-annotation of \( \S \) and \( \Phi \) a relevant part of \( \Theta \).
	A state \( s \in \S \) is a \emph{limit state} of \( \Phi \) if \( \nabla \Gamma^\lambda \in \Phi_s \) for some \( \Gamma \) and limit ordinal \( \lambda \). 
	A pair of states $( s , t )$ in \( \S \) is a \emph{repetition pair} if:
\begin{enumerate}
	\item there is a path \( (s_i)_{i\le k} \) with \( s = s_0 \) and \( t = s_k \),
	\item $( \Theta_{s}^- , \Phi_{s}^- ) = ( \Theta_{t}^- , \Phi_{t}^- ) $,
	\item\label{reptrace-3} \( s \) and \( t \) are limit states and \( \nabla \Gamma^\alpha \in \Phi_s \) and \( \nabla \Gamma^\beta \in \Phi_{t} \) for some \( \Gamma \) and \( \beta < \alpha \).
\end{enumerate}
We refer to $t$ as the \emph{bud} and $s$ as the \emph{companion} of the repetition pair.
\end{definition}
Call a path whose limit states are all optimal an optimal path. 
\begin{lemma}\label{optimal-trace-points}
On every optimal path there are no repetition pairs.
\end{lemma}
\begin{proof}
Straightforward from the fact that the bud node is a non-optimal limit point by definition.
\end{proof}
Finally, the next proposition specifies sufficient conditions for the existence of repetition pairs.
\begin{proposition}\label{pairexists}
	Let $ N = 2^{ 2 \card \phi } + 1 $ and $\Theta$ a conservative well-annotation of \( \S \) with respect to \( \phi \). 
	Let \( (s_i)_{i\le k} \) be a path through \( \S \) and \( ( \phi_i , \alpha_i )_{i \le k} \) a sequence of annotated formulas such that \( \phi_i^{\alpha_i} \in \Phi_{s_i} \) for each \( i \).
	If \( \alpha_0 \ge \omega.N \) and \( \alpha_k = 0 \), then some \( (s_i,s_j) \) is a repetition pair.
\end{proposition}
\begin{proof}
	\Cref{d-relevant} ensures that on every sufficiently long path in which the relevant part remains non-empty, there are $2^{2 \card \phi }$ limit states with the corresponding limit ordinal strictly decreasing between states.
	As there are \( 2^{\card \phi} \) subsets of \( \FL(\phi) \) the existence of a repetition pair is immediate.
\end{proof}
Consider a state $\Theta_s$ in $\Theta$ with relevant part $\Phi_s$ such that $\nabla \Gamma^\lambda \in \Phi_s$. If $O(\nabla \Gamma,\Theta_s^-)\ge \lambda + \omega$, by definition there is a conservative well-annotation $\Theta'$ such that ${\Theta}^-_{s} = (\Theta'_{\rho})^-$ and $\nabla \Gamma^\beta \in \Theta'_{\rho}$ for some $\beta \ge  \lambda + \omega$. We call \emph{pumping} of $\Theta_s$  the operation of replacing in $\Theta$ the branch rooted at $\Theta_s$ with the conservative well-annotation $\Theta'$. 
\begin{restatable}[First pumping lemma]{lemma}{firstpumping}
\label{pump}
There exists $ N < \omega $ such that for all \( \Gamma \subseteq \FL(X,E) \) and \( x \in X \), if \( O( x, \Gamma ) < \omega^2\) then \( O( x, \Gamma ) \le \omega.N \).
\end{restatable}

\begin{proof}[Proof sketch]
  We prove the contrapositive statement, i.e., if \( O( x, \Gamma ) > \omega.N \) then \( O( x, \Gamma ) \ge \omega^2 \) for suitable \( N \). Let $ N =2^{ 2 \card \phi } + 1 $ and assume a formula \( (X,E)\) and a \(\Theta\) such that \(O( x, \Theta_\rho  )=\kappa\) for some \(\omega.N< \kappa < \omega^2\). We prove that the root of such $\Theta$ cannot be optimal, i.e. that \( O(x,\Theta_\rho )>\kappa\), by showing that that would entail the existence of another conservative well-annotation $\Theta'$ with an optimal path and a repetition pair, contradicting \cref{optimal-trace-points}. From the fact that $\Theta_\rho$ is assumed optimal, we proceed by pumping every successor $r\in R[\rho]$ that is not optimal (wrt the unique relevant $\nabla \Gamma_r$). Since by definition there must be at least one optimal $s\in R[\rho]$, we move to all optimal successors and repeat the pumping of their non-optimal successors. The conservative well-annotation $\Theta'$ obtained at the end of this process by definition has an optimal path that is strictly decreasing, hence fulfilling the conditions of \cref{pairexists}, from which we obtain the contradiction with \cref{optimal-trace-points}. It follows that \(O( x, \Theta_\rho ) = \kappa >\omega.N\) entails $\kappa \ge \omega^2$. 
\end{proof}
The first pumping lemma eliminates ordinals sufficiently close to $\omega^2$ as being closure ordinals of \( \Sigma \)-formulas of bounded size.
In the rest of the argument we do the same for ordinals between $\omega^2$ and $\omega_1$.
\begin{restatable}[Second pumping lemma]{lemma}{slump}\label{slump}
	For all \( \kappa < \omega_1 \), \( \Gamma \subseteq \FL(X,E) \) and \( x \in X \), if \( O( x , \Gamma ) \le \kappa \) then \( O( x , \Gamma ) < \omega^2 \).
\end{restatable}
\begin{proof}[Proof sketch]
    The proof proceeds by transfinite induction on \( \kappa \) and is similar in spirit to the first pumping lemma.
This second lemma doesn't rely directly on optimality because a greater generality is needed, but \cref{pump} serves as base case in the argument for every $\kappa \ge \omega^2$.
Given a conservative well-annotation \( \Theta \) with \( x^\kappa \in \Theta_\rho \) for \( \kappa \ge \omega^2 \) a series of substitutions to the underlying tree induces a tree that satisfies \( x \) at the root and for which all well-annotations necessitate a strictly larger annotation of this variable.
In the case that \( \kappa > \omega^2 \) is a successor ordinal, the substitutions can be performed directly to the successors of the root by appealing to the induction hypothesis.
The case of a limit ordinal is more involved and requires identifying, via the relevant part of \( \Theta \), transfinitely many candidate states at which the induction hypothesis can be applied.
After performing the substitutions, an infinite descent argument establishes a necessary increase in the ordinal annotation.
\end{proof}
\begin{theorem}
	A countable ordinal \( \alpha \) is the closure ordinal of a formula in the \( \Sigma \)-fragment iff \( \alpha < \omega^2 \).
\end{theorem}  
\begin{proof}
	One direction is provided by \cref{Czarnecki}. 
	For the other direction, let \( \phi \) be a \( \Sigma \)-formula with \( \CO \phi < \omega_1 \). 
	By \cref{t-conj-equiv} we may assume \( \phi \) is conjunctive.
	Thus, given \( x \) the initial variable of \( \phi \), we have \(\CO\phi=\alpha \) entails \(O( x , \Delta ) = \alpha \) for some $\Delta \subseteq \FL(X,E)$,	whence \( \CO \phi < \omega^2 \) by \cref{slump}.
\end{proof}

\section{Conclusion}\label{conclusion}
 \enlargethispage{\baselineskip}
We have shown that the countable closure ordinals of formulas in the $\Sigma$-fragment are strictly bounded by $\omega^2$.
The result extends what was claimed in \cite{AfshariL13} and is obtained using a different method that circumvents the shortcomings of the approach taken there. 
The main ingredient introduced is a reworked version of well-annotation from \cite{Koz88}, with which the focus is cast directly on transforming frames.

The machinery described in this article is for most part independent of the $\Sigma$-fragment to which they are applied. 
An immediate continuation of this work is, therefore, to examine the versatility of the tools for investigating closure ordinal of $\mu x\,\phi$ where $\phi$ is any $\mu$-calculus formula.
Another research direction, suggested by some of the insights from Section \ref{sec:CO}, is  to directly study the relation between closure ordinals and semantic equivalence, i.e., the syntactic operations which preserve closure ordinals.

\bibliographystyle{eptcs}
\bibliography{FICS.bib}

\begin{thebibliography}{10}
\providecommand{\bibitemdeclare}[2]{}
\providecommand{\surnamestart}{}
\providecommand{\surnameend}{}
\providecommand{\urlprefix}{Available at }
\providecommand{\url}[1]{\texttt{#1}}
\providecommand{\href}[2]{\texttt{#2}}
\providecommand{\urlalt}[2]{\href{#1}{#2}}
\providecommand{\doi}[1]{doi:\urlalt{https://doi.org/#1}{#1}}
\providecommand{\eprint}[1]{arXiv:\urlalt{https://arxiv.org/abs/#1}{#1}}
\providecommand{\bibinfo}[2]{#2}

\bibitemdeclare{inproceedings}{AJL19}
\bibitem{AJL19}
\bibinfo{author}{Bahareh \surnamestart Afshari\surnameend},
  \bibinfo{author}{Gerhard \surnamestart J{\"{a}}ger\surnameend} \&
  \bibinfo{author}{Graham~E. \surnamestart Leigh\surnameend}
  (\bibinfo{year}{2019}): \emph{\bibinfo{title}{An Infinitary Treatment of Full
  Mu-Calculus}}.
\newblock In \bibinfo{editor}{Rosalie \surnamestart Iemhoff\surnameend},
  \bibinfo{editor}{Michael \surnamestart Moortgat\surnameend} \&
  \bibinfo{editor}{Ruy J. G.~B. \surnamestart de~Queiroz\surnameend}, editors:
  {\slshape \bibinfo{booktitle}{Logic, Language, Information, and Computation -
  26th International Workshop, WoLLIC 2019, Utrecht, The Netherlands, July 2-5,
  2019, Proceedings}}, {\slshape \bibinfo{series}{Lecture Notes in Computer
  Science}} \bibinfo{volume}{11541}, \bibinfo{publisher}{Springer}, pp.
  \bibinfo{pages}{17--34}, \doi{10.1007/978-3-662-59533-6\_2}.

\bibitemdeclare{unpublished}{AL16}
\bibitem{AL16}
\bibinfo{author}{Bahareh \surnamestart Afshari\surnameend} \&
  \bibinfo{author}{Graham~E. \surnamestart Leigh\surnameend}:
  \emph{\bibinfo{title}{Closure Ordinals: Revisions and Proofs}}.
\newblock
  \urlprefix\url{https://surfdrive.surf.nl/files/index.php/s/AWpStp9s2SynBno}.
\newblock \bibinfo{note}{Unpublished notes}.

\bibitemdeclare{inproceedings}{AfshariL13}
\bibitem{AfshariL13}
\bibinfo{author}{Bahareh \surnamestart Afshari\surnameend} \&
  \bibinfo{author}{Graham~E. \surnamestart Leigh\surnameend}
  (\bibinfo{year}{2013}): \emph{\bibinfo{title}{On closure ordinals for the
  modal mu-calculus}}.
\newblock In \bibinfo{editor}{Simona Ronchi~Della \surnamestart
  Rocca\surnameend}, editor: {\slshape \bibinfo{booktitle}{Computer Science
  Logic 2013 {(CSL} 2013), {CSL} 2013, September 2-5, 2013, Torino, Italy}},
  {\slshape \bibinfo{series}{LIPIcs}}~\bibinfo{volume}{23},
  \bibinfo{publisher}{Schloss Dagstuhl - Leibniz-Zentrum f{\"{u}}r Informatik},
  pp. \bibinfo{pages}{30--44}, \doi{10.4230/LIPIcs.CSL.2013.30}.

\bibitemdeclare{article}{Ong21}
\bibitem{Ong21}
\bibinfo{author}{Christopher~H. \surnamestart Broadbent\surnameend},
  \bibinfo{author}{Arnaud \surnamestart Carayol\surnameend},
  \bibinfo{author}{C.{-}H.~Luke \surnamestart Ong\surnameend} \&
  \bibinfo{author}{Olivier \surnamestart Serre\surnameend}
  (\bibinfo{year}{2021}): \emph{\bibinfo{title}{Higher-order Recursion Schemes
  and Collapsible Pushdown Automata: Logical Properties}}.
\newblock {\slshape \bibinfo{journal}{{ACM} Trans. Comput. Log.}}
  \bibinfo{volume}{22}(\bibinfo{number}{2}), pp. \bibinfo{pages}{12:1--12:37},
  \doi{10.1145/3452917}.

\bibitemdeclare{inproceedings}{Nir15}
\bibitem{Nir15}
\bibinfo{author}{Pablo~F. \surnamestart Castro\surnameend},
  \bibinfo{author}{Cecilia \surnamestart Kilmurray\surnameend} \&
  \bibinfo{author}{Nir \surnamestart Piterman\surnameend}
  (\bibinfo{year}{2015}): \emph{\bibinfo{title}{Tractable Probabilistic
  mu-Calculus That Expresses Probabilistic Temporal Logics}}.
\newblock In \bibinfo{editor}{Ernst~W. \surnamestart Mayr\surnameend} \&
  \bibinfo{editor}{Nicolas \surnamestart Ollinger\surnameend}, editors:
  {\slshape \bibinfo{booktitle}{32nd International Symposium on Theoretical
  Aspects of Computer Science, {STACS} 2015, March 4-7, 2015, Garching,
  Germany}}, {\slshape \bibinfo{series}{LIPIcs}}~\bibinfo{volume}{30},
  \bibinfo{publisher}{Schloss Dagstuhl - Leibniz-Zentrum f{\"{u}}r Informatik},
  pp. \bibinfo{pages}{211--223}, \doi{10.4230/LIPIcs.STACS.2015.211}.

\bibitemdeclare{article}{Cz10}
\bibitem{Cz10}
\bibinfo{author}{Marek \surnamestart Czarnecki\surnameend}
  (\bibinfo{year}{2010}): \emph{\bibinfo{title}{How fast can the fixpoints in
  modal mu-calculus be reached}}.
\newblock {\slshape \bibinfo{journal}{Fixed Points in Computer Science}}, pp.
  \bibinfo{pages}{35--39}.

\bibitemdeclare{incollection}{DGL16}
\bibitem{DGL16}
\bibinfo{author}{St{\'e}phane \surnamestart Demri\surnameend},
  \bibinfo{author}{Valentin \surnamestart Goranko\surnameend} \&
  \bibinfo{author}{Martin \surnamestart Lange\surnameend}
  (\bibinfo{year}{2016}): \emph{\bibinfo{title}{The Modal Mu-Calculus}}.
\newblock In: {\slshape \bibinfo{booktitle}{Temporal Logics in Computer
  Science: Finite-State Systems}}, \bibinfo{series}{Cambridge Tracts in
  Theoretical Computer Science}, \bibinfo{publisher}{Cambridge University
  Press}, pp. \bibinfo{pages}{271--328}, \doi{10.1017/CBO9781139236119.008}.

\bibitemdeclare{article}{FischerLadner-79}
\bibitem{FischerLadner-79}
\bibinfo{author}{Michael~J. \surnamestart Fischer\surnameend} \&
  \bibinfo{author}{Richard~E. \surnamestart Ladner\surnameend}
  (\bibinfo{year}{1979}): \emph{\bibinfo{title}{Propositional dynamic logic of
  regular programs}}.
\newblock {\slshape \bibinfo{journal}{Journal of Computer and System Sciences}}
  \bibinfo{volume}{18}(\bibinfo{number}{2}), pp. \bibinfo{pages}{194--211},
  \doi{10.1016/0022-0000(79)90046-1}.

\bibitemdeclare{inproceedings}{Fon08}
\bibitem{Fon08}
\bibinfo{author}{Ga{\"e}lle \surnamestart Fontaine\surnameend}
  (\bibinfo{year}{2008}): \emph{\bibinfo{title}{Continuous Fragment of the
  mu-Calculus}}.
\newblock In \bibinfo{editor}{Michael \surnamestart Kaminski\surnameend} \&
  \bibinfo{editor}{Simone \surnamestart Martini\surnameend}, editors: {\slshape
  \bibinfo{booktitle}{Computer Science Logic}}, \bibinfo{publisher}{Springer
  Berlin Heidelberg}, pp. \bibinfo{pages}{139--153},
  \doi{10.1007/978-3-540-87531-4_12}.

\bibitemdeclare{phdthesis}{Fon10}
\bibitem{Fon10}
\bibinfo{author}{Ga{\"e}lle \surnamestart Fontaine\surnameend}
  (\bibinfo{year}{2010}): \emph{\bibinfo{title}{Modal fixpoint logic: some
  model theoretic questions}}.
\newblock Ph.D. thesis, \bibinfo{school}{University of Amsterdam}.

\bibitemdeclare{article}{FV18}
\bibitem{FV18}
\bibinfo{author}{Ga{\"e}lle \surnamestart Fontaine\surnameend} \&
  \bibinfo{author}{Yde \surnamestart Venema\surnameend} (\bibinfo{year}{2018}):
  \emph{\bibinfo{title}{Some model theory for the modal mu-calculus: syntactic
  characterisations of semantic properties}}.
\newblock {\slshape \bibinfo{journal}{Logical Methods in Computer Science;
  Volume 14}}, p. \bibinfo{pages}{Issue 1}, \doi{10.23638/LMCS-14(1:14)2018}.

\bibitemdeclare{article}{GS18}
\bibitem{GS18}
\bibinfo{author}{Maria~Jo{\~{a}}o \surnamestart Gouveia\surnameend} \&
  \bibinfo{author}{Luigi \surnamestart Santocanale\surnameend}
  (\bibinfo{year}{2019}): \emph{\bibinfo{title}{{\(\aleph\)}\({}_{\mbox{1}}\)
  and the modal {\(\mu\)}-calculus}}.
\newblock {\slshape \bibinfo{journal}{Log. Methods Comput. Sci.}}
  \bibinfo{volume}{15}(\bibinfo{number}{4}), \doi{10.23638/LMCS-15(4:1)2019}.

\bibitemdeclare{book}{green}
\bibitem{green}
\bibinfo{editor}{Erich \surnamestart Gr\"adel\surnameend},
  \bibinfo{editor}{Wolfgang \surnamestart Thomas\surnameend} \&
  \bibinfo{editor}{Thomas \surnamestart Wilke\surnameend}, editors
  (\bibinfo{year}{2002}): \emph{\bibinfo{title}{Automata, Logics, and Infinite
  Games: {A} Guide to Current Research}}.
\newblock \bibinfo{publisher}{Springer Berlin, Heidelberg},
  \doi{10.1007/3-540-36387-4}.

\bibitemdeclare{article}{GutsfeldMO21}
\bibitem{GutsfeldMO21}
\bibinfo{author}{Jens~Oliver \surnamestart Gutsfeld\surnameend},
  \bibinfo{author}{Markus \surnamestart M{\"{u}}ller{-}Olm\surnameend} \&
  \bibinfo{author}{Christoph \surnamestart Ohrem\surnameend}
  (\bibinfo{year}{2021}): \emph{\bibinfo{title}{Automata and fixpoints for
  asynchronous hyperproperties}}.
\newblock {\slshape \bibinfo{journal}{Proc. {ACM} Program. Lang.}}
  \bibinfo{volume}{5}(\bibinfo{number}{{POPL}}), pp. \bibinfo{pages}{1--29},
  \doi{10.1145/3434319}.

\bibitemdeclare{incollection}{JW95}
\bibitem{JW95}
\bibinfo{author}{David \surnamestart Janin\surnameend} \& \bibinfo{author}{Igor
  \surnamestart Walukiewicz\surnameend} (\bibinfo{year}{1995}):
  \emph{\bibinfo{title}{Automata for the modal mu-calculus and related
  results}}.
\newblock In \bibinfo{editor}{Ji{\v r}{\'\i} \surnamestart
  Wiedermann\surnameend} \& \bibinfo{editor}{Petr \surnamestart
  H{\'a}jek\surnameend}, editors: {\slshape \bibinfo{booktitle}{Mathematical
  Foundations of Computer Science 1995}}, {\slshape \bibinfo{series}{Lecture
  Notes in Computer Science}} \bibinfo{volume}{969},
  \bibinfo{publisher}{Springer Berlin Heidelberg}, pp.
  \bibinfo{pages}{552--562}, \doi{10.1007/3-540-60246-1_160}.

\bibitemdeclare{article}{Koz88}
\bibitem{Koz88}
\bibinfo{author}{Dexter \surnamestart Kozen\surnameend} (\bibinfo{year}{1988}):
  \emph{\bibinfo{title}{A Finite Model Theorem for the Propositional
  mu-Calculus}}.
\newblock {\slshape \bibinfo{journal}{Studia Logica: An International Journal
  for Symbolic Logic}} \bibinfo{volume}{47}(\bibinfo{number}{3}), pp.
  \bibinfo{pages}{233--241}, \doi{10.1007/BF00370554}.

\bibitemdeclare{phdthesis}{Kretz06}
\bibitem{Kretz06}
\bibinfo{author}{Mathis \surnamestart Kretz\surnameend} (\bibinfo{year}{2006}):
  \emph{\bibinfo{title}{Proof-theoretic aspects of modal logic with fixed
  points}}.
\newblock Ph.D. thesis, \bibinfo{school}{University of Bern}.

\bibitemdeclare{article}{Kupke12}
\bibitem{Kupke12}
\bibinfo{author}{Clemens \surnamestart Kupke\surnameend},
  \bibinfo{author}{Alexander \surnamestart Kurz\surnameend} \&
  \bibinfo{author}{Yde \surnamestart Venema\surnameend} (\bibinfo{year}{2012}):
  \emph{\bibinfo{title}{Completeness for the coalgebraic cover modality}}.
\newblock {\slshape \bibinfo{journal}{Log. Methods Comput. Sci.}}
  \bibinfo{volume}{8}(\bibinfo{number}{3}), \doi{10.2168/LMCS-8(3:2)2012}.

\bibitemdeclare{article}{KupkeP11}
\bibitem{KupkeP11}
\bibinfo{author}{Clemens \surnamestart Kupke\surnameend} \&
  \bibinfo{author}{Dirk \surnamestart Pattinson\surnameend}
  (\bibinfo{year}{2011}): \emph{\bibinfo{title}{Coalgebraic semantics of modal
  logics: An overview}}.
\newblock {\slshape \bibinfo{journal}{Theor. Comput. Sci.}}
  \bibinfo{volume}{412}(\bibinfo{number}{38}), pp. \bibinfo{pages}{5070--5094},
  \doi{10.1016/j.tcs.2011.04.023}.

\bibitemdeclare{inproceedings}{Lange15}
\bibitem{Lange15}
\bibinfo{author}{Martin \surnamestart Lange\surnameend} (\bibinfo{year}{2015}):
  \emph{\bibinfo{title}{The Arity Hierarchy in the Polyadic
  {\(\mu\)}-Calculus}}.
\newblock In \bibinfo{editor}{Ralph \surnamestart Matthes\surnameend} \&
  \bibinfo{editor}{Matteo \surnamestart Mio\surnameend}, editors: {\slshape
  \bibinfo{booktitle}{Proceedings Tenth International Workshop on Fixed Points
  in Computer Science, {FICS} 2015, Berlin, Germany, September 11-12, 2015}},
  {\slshape \bibinfo{series}{{EPTCS}}} \bibinfo{volume}{191}, pp.
  \bibinfo{pages}{105--116}, \doi{10.4204/EPTCS.191.10}.

\bibitemdeclare{inproceedings}{LiuS15}
\bibitem{LiuS15}
\bibinfo{author}{Wanwei \surnamestart Liu\surnameend}, \bibinfo{author}{Lei
  \surnamestart Song\surnameend}, \bibinfo{author}{Ji~\surnamestart
  Wang\surnameend} \& \bibinfo{author}{Lijun \surnamestart Zhang\surnameend}
  (\bibinfo{year}{2015}): \emph{\bibinfo{title}{A Simple Probabilistic
  Extension of Modal Mu-calculus}}.
\newblock In \bibinfo{editor}{Qiang \surnamestart Yang\surnameend} \&
  \bibinfo{editor}{Michael~J. \surnamestart Wooldridge\surnameend}, editors:
  {\slshape \bibinfo{booktitle}{Proceedings of the Twenty-Fourth International
  Joint Conference on Artificial Intelligence, {IJCAI} 2015, Buenos Aires,
  Argentina, July 25-31, 2015}}, \bibinfo{publisher}{{AAAI} Press}, pp.
  \bibinfo{pages}{882--888}, \doi{10.48550/arXiv.1504.07737}.

\bibitemdeclare{mastersthesis}{Milanese18}
\bibitem{Milanese18}
\bibinfo{author}{Gian~Carlo \surnamestart Milanese\surnameend}
  (\bibinfo{year}{2018}): \emph{\bibinfo{title}{An exploration of closure
  ordinals in the modal $\mu$-calculus}}.
\newblock \bibinfo{type}{Master thesis}, \bibinfo{school}{University of
  Amsterdam}.

\bibitemdeclare{inproceedings}{MilaneseV19}
\bibitem{MilaneseV19}
\bibinfo{author}{Gian~Carlo \surnamestart Milanese\surnameend} \&
  \bibinfo{author}{Yde \surnamestart Venema\surnameend} (\bibinfo{year}{2019}):
  \emph{\bibinfo{title}{Closure Ordinals of the Two-Way Modal
  {\(\mathrm{\mu}\)}-Calculus}}.
\newblock In \bibinfo{editor}{Rosalie \surnamestart Iemhoff\surnameend},
  \bibinfo{editor}{Michael \surnamestart Moortgat\surnameend} \&
  \bibinfo{editor}{Ruy J. G.~B. \surnamestart de~Queiroz\surnameend}, editors:
  {\slshape \bibinfo{booktitle}{Logic, Language, Information, and Computation -
  26th International Workshop, WoLLIC 2019, Utrecht, The Netherlands, July 2-5,
  2019, Proceedings}}, {\slshape \bibinfo{series}{Lecture Notes in Computer
  Science}} \bibinfo{volume}{11541}, \bibinfo{publisher}{Springer}, pp.
  \bibinfo{pages}{498--515}, \doi{10.1007/978-3-662-59533-6\_30}.

\bibitemdeclare{article}{Mio17}
\bibitem{Mio17}
\bibinfo{author}{Matteo \surnamestart Mio\surnameend} \& \bibinfo{author}{Alex
  \surnamestart Simpson\surnameend} (\bibinfo{year}{2017}):
  \emph{\bibinfo{title}{{\L}ukasiewicz {\(\mu\)}-calculus}}.
\newblock {\slshape \bibinfo{journal}{Fundam. Informaticae}}
  \bibinfo{volume}{150}(\bibinfo{number}{3-4}), pp. \bibinfo{pages}{317--346},
  \doi{10.3233/FI-2017-1472}.

\bibitemdeclare{article}{NW96}
\bibitem{NW96}
\bibinfo{author}{Damian \surnamestart Niwi{\'n}ski\surnameend} \&
  \bibinfo{author}{Igor \surnamestart Walukiewicz\surnameend}
  (\bibinfo{year}{1996}): \emph{\bibinfo{title}{Games for the mu-calculus}}.
\newblock {\slshape \bibinfo{journal}{Theoretical Computer Science}}
  \bibinfo{volume}{163}(\bibinfo{number}{1}), pp. \bibinfo{pages}{99--116},
  \doi{10.1016/0304-3975(95)00136-0}.

\bibitemdeclare{article}{NW03}
\bibitem{NW03}
\bibinfo{author}{Damian \surnamestart Niwi{\'n}ski\surnameend} \&
  \bibinfo{author}{Igor \surnamestart Walukiewicz\surnameend}
  (\bibinfo{year}{2003}): \emph{\bibinfo{title}{A gap property of deterministic
  tree languages}}.
\newblock {\slshape \bibinfo{journal}{Theor. Comput. Sci.}}
  \bibinfo{volume}{303}(\bibinfo{number}{1}), pp. \bibinfo{pages}{215--231},
  \doi{10.1016/S0304-3975(02)00452-8}.

\bibitemdeclare{inproceedings}{Ong06}
\bibitem{Ong06}
\bibinfo{author}{{C.{-}H. Luke} \surnamestart Ong\surnameend}
  (\bibinfo{year}{2006}): \emph{\bibinfo{title}{On Model-Checking Trees
  Generated by Higher-Order Recursion Schemes}}.
\newblock In: {\slshape \bibinfo{booktitle}{21th {IEEE} Symposium on Logic in
  Computer Science {(LICS} 2006), 12-15 August 2006, Seattle, WA, USA,
  Proceedings}}, \bibinfo{publisher}{{IEEE} Computer Society}, pp.
  \bibinfo{pages}{81--90}, \doi{10.1109/LICS.2006.38}.

\bibitemdeclare{article}{Otto99}
\bibitem{Otto99}
\bibinfo{author}{Martin \surnamestart Otto\surnameend} (\bibinfo{year}{1999}):
  \emph{\bibinfo{title}{Bisimulation-invariant {PTIME} and higher-dimensional
  {\(\mathrm{\mu}\)}-calculus}}.
\newblock {\slshape \bibinfo{journal}{Theor. Comput. Sci.}}
  \bibinfo{volume}{224}(\bibinfo{number}{1-2}), pp. \bibinfo{pages}{237--265},
  \doi{10.1016/S0304-3975(98)00314-4}.

\bibitemdeclare{inproceedings}{SW16}
\bibitem{SW16}
\bibinfo{author}{Michal \surnamestart Skrzypczak\surnameend} \&
  \bibinfo{author}{Igor \surnamestart Walukiewicz\surnameend}
  (\bibinfo{year}{2016}): \emph{\bibinfo{title}{Deciding the Topological
  Complexity of {B{\"u}chi} Languages}}.
\newblock In \bibinfo{editor}{Ioannis \surnamestart
  Chatzigiannakis\surnameend}, \bibinfo{editor}{Michael \surnamestart
  Mitzenmacher\surnameend}, \bibinfo{editor}{Yuval \surnamestart
  Rabani\surnameend} \& \bibinfo{editor}{Davide \surnamestart
  Sangiorgi\surnameend}, editors: {\slshape \bibinfo{booktitle}{43rd
  International Colloquium on Automata, Languages, and Programming (ICALP
  2016)}}, {\slshape \bibinfo{series}{Leibniz International Proceedings in
  Informatics (LIPIcs)}}~\bibinfo{volume}{55}, \bibinfo{publisher}{Schloss
  Dagstuhl--Leibniz-Zentrum fuer Informatik}, \bibinfo{address}{Dagstuhl,
  Germany}, pp. \bibinfo{pages}{99:1--99:13},
  \doi{10.4230/LIPIcs.ICALP.2016.99}.

\bibitemdeclare{phdthesis}{MS16}
\bibitem{MS16}
\bibinfo{author}{Michał \surnamestart Skrzypczak\surnameend}
  (\bibinfo{year}{2014}): \emph{\bibinfo{title}{Descriptive set theoretic
  methods in automata theory}}.
\newblock \bibinfo{type}{{PhD}\ thesis}, \bibinfo{school}{University of
  Warsaw}, \doi{10.1007/978-3-662-52947-8}.

\bibitemdeclare{article}{SE89}
\bibitem{SE89}
\bibinfo{author}{Robert~S. \surnamestart Streett\surnameend} \&
  \bibinfo{author}{E.~Allen \surnamestart Emerson\surnameend}
  (\bibinfo{year}{1989}): \emph{\bibinfo{title}{An automata theoretic decision
  procedure for the propositional mu-calculus}}.
\newblock {\slshape \bibinfo{journal}{Information and Computation}}
  \bibinfo{volume}{81}(\bibinfo{number}{3}), pp. \bibinfo{pages}{249--264},
  \doi{10.1016/0890-5401(89)90031-X}.

\bibitemdeclare{article}{Wal00}
\bibitem{Wal00}
\bibinfo{author}{Igor \surnamestart Walukiewicz\surnameend}
  (\bibinfo{year}{2000}): \emph{\bibinfo{title}{Completeness of {Kozen's}
  Axiomatisation of the Propositional mu-Calculus}}.
\newblock {\slshape \bibinfo{journal}{Information and Computation}}
  \bibinfo{volume}{157}(\bibinfo{number}{1}), pp. \bibinfo{pages}{142--182},
  \doi{10.1006/inco.1999.2836}.

\end{thebibliography}

\end{document}